\documentclass{article}
\usepackage{dcolumn}% Align table columns on decimal point
\usepackage{bm}% bold math

\usepackage{verbatim}       % Defines \begin{comment}, \end{comment}

\usepackage[dvips]{graphicx}
\usepackage{amsthm}
\usepackage{amssymb}
\usepackage{bm}
\usepackage{amstext}
\usepackage{psfrag}
\usepackage{amsmath}
\usepackage{amsfonts}
\usepackage{mathrsfs}
\usepackage{cite}

\newcommand{\p}{\partial}

\newcommand{\dd}{{\rm d}}
\newcommand{\bd}{\begin{definition}}                %inizia definizione
\newcommand{\ed}{\end{definition}}                  %fine definizione
\newcommand{\bc}{\begin{corollary}}                 %inizia corollario
\newcommand{\ec}{\end{corollary}}                   %fine corollario
\newcommand{\bl}{\begin{lemma}}                     %inizia lemma
\newcommand{\el}{\end{lemma}}                       %fine lemma
\newcommand{\bp}{\begin{proposition}}            %inizia proposizione
\newcommand{\ep}{\end{proposition}}                %fine proposizione
\newcommand{\bere}{\begin{remark}}                  %inizia osservazione
\newcommand{\ere}{\end{remark}}                     %fine oservazione

\newcommand{\bt}{\begin{theorem}}
\newcommand{\et}{\end{theorem}}

\newcommand{\be}{\begin{equation}}
\newcommand{\ee}{\end{equation}}

\newcommand{\bit}{\begin{itemize}}
\newcommand{\eit}{\end{itemize}}
\newtheorem{theorem}{Theorem}[section]
\newtheorem{corollary}[theorem]{Corollary}
\newtheorem{lemma}[theorem]{Lemma}
\newtheorem{proposition}[theorem]{Proposition}
\theoremstyle{definition}
\newtheorem{definition}[theorem]{Definition}
\theoremstyle{remark}
\newtheorem{remark}[theorem]{Remark}
\newtheorem{example}[theorem]{Example}

\begin{document}
%
%\DeclareGraphicsExtensions{.pdf}

%\title{Connection between Lorentzian distance and mechanical least action in spacetimes admitting a parallel null vector}

\title{Completeness of first and second order ODE flows and of Euler-Lagrange equations}

\author{E. Minguzzi\footnote{Dipartimento di Matematica e Informatica ``U. Dini'', Universit\`a degli Studi di Firenze,  Via
S. Marta 3,  I-50139 Firenze, Italy. E-mail:
ettore.minguzzi@unifi.it}}

%\pacs{04.20.Gz, 04.30.-w}

\date{}
\maketitle

\begin{abstract}
\noindent Two  results on the completeness of maximal solutions to
first and second order ordinary differential equations (or
inclusions) over complete Riemannian manifolds, with possibly
time-dependent metrics, are obtained. Applications to Lagrangian
mechanics and gravitational waves are given.
\end{abstract}

%\noindent Key Words:

\section{Introduction}

In this work we shall obtain some completeness result for maximal
solutions to ordinary first and second order equations over a
complete Riemannian manifold $(Q,a)$. We are concerned with the
equations
\begin{align}
\dot{q}&=\nu(t,q) \label{nia} \\
\frac{D }{\dd t}\, \dot{q}&=f(t,q(t), \dot{q}(t)), \label{nio}
\end{align}
where $D$ is the affine connection of $(Q,a)$, $\nu\colon
\mathbb{R}\times Q\to TQ$ is a time-dependent vector field and
$f\colon \mathbb{R} \times TQ\to TQ$ is a time-dependent and velocity
dependent vector field.

For $Q=\mathbb{R}^n$ and Eq.\ (\ref{nia}) the problem is answered
satisfactorily by Wintner's theorem \cite{wintner45,wintner46}
\cite[Theor.\ 5.1, Chap.\ 3]{hartman64} on the completeness of {\em
first order} flows on $\mathbb{R}^n$, and by its refinements and
variations
\cite{wintner46b,cooke55,conti56,conti56b,strauss67,burton77,constantin91,gliklikh06}.
This solution implies a straightforward answer to the analogous
completeness problem for Eq.~(\ref{nio}), indeed, this second order
equation  can be rewritten as a first order equation on $TQ$, which
for $Q=\mathbb{R}^n$ is diffeomorphic to $\mathbb{R}^{2n}$. In any
case there have been other studies of Eq.\ (\ref{nio}) in
$\mathbb{R}^n$ and especially in $\mathbb{R}$. For instance, Hartman
and Wintner \cite{hartman51,mustafa02} \cite[Theor. 5.2, Chap.\
12]{hartman64}, taking advantage of some estimates by Bernstein and
Nagumo, showed that if $f$ has an almost quadratic asymptotic
dependence on $v$, and satisfies some other assumptions, then a
forward complete solution exists which starts from any prescribed
initial point $q(0)=q_0$. However, in their theorem they could not
impose an a priori  value for $\dot{q}(0)$. This limitation can be
regarded as a consequence of the fact that the quadratic dependence
on the velocity is just at the verge of spoiling completeness
(Example \ref{cxe}).

%For $Q\ne \mathbb{R}^n$
%
%
%The autonomous case for force fields $f$ of gradient type was
%considered in \cite{ebin70,weinstein70,gordon70}, while the non
%autonomous case for slightly more general fields, with an additional
%affine dependence on the velocity, has been  considered in
%\cite{candela12}. These works do not mention Wintner's theorem
%\cite{wintner45} \cite[Theor.\ 5.1, Chap.\ 3]{hartman64} on the
%completeness of {\em first order} flows on $\mathbb{R}^n$. Wintner
%result was followed by several refinements and variations
%\cite{wintner46,cooke55,conti56,conti56b,strauss67,burton77,constantin91,gliklikh06}.
%Clearly, the second order equation (\ref{nio}) can be rewritten as a
%first order equation on $TQ$. Now,

If $Q$ is curved we still have that every coordinate chart
$\{q^\alpha\}$ on an open set $U\subset Q$ induces a chart
$\{q^\alpha, \dot{q}^\beta\}$ on $\pi^{-1}(U)\subset TQ$, where
$\pi: TQ\to Q$ is the usual projection. Thus one can tackle  the
problem of completeness to maximal solutions to the second order
differential equation (\ref{nio}) by studying the completeness of a
flow on  open sets of $\mathbb{R}^n$. Wintner's theorem assures that
under suitable assumptions the solution will be either complete or
will reach the boundary of the coordinate patch (i.e.\  $q(t)$ will
reach the boundary of $U$). In the latter case Wintner's theorem can
be applied once again, so that, inductively, either the solution is
complete or it crosses infinite coordinate charts in finite time.
The application of Wintner's theorem to the general Riemannian case
is not completely straightforward since one would like to formulate
the hypothesis in a natural coordinate-free language, for instance
expressing Wintner's conditions  on the asymptotic behavior of $f$
in terms of functions constructed from the Riemannian distance.

The problem of the completeness of maximal solutions to Eq.\
(\ref{nio}) in Riemannian spaces was considered by other authors
who, apparently unaware of Winter's work, passed through similar
ideas. The autonomous case for force fields $f$ of gradient type,
was considered in \cite{ebin70,weinstein70,gordon70}, while the non
autonomous case for more general fields, with an additional affine
dependence on the velocity, has been  recently considered in
\cite{candela13}. The results obtained in these studies  do not seem
to reach the generality that one would expect from the analogy with
Wintner's results in $\mathbb{R}^n$, particularly in connection with
the possible superlinear growth of the fields $\nu$ and $f$, or of
the dependence of $f$ on the velocity. In this work we shall remedy
this situation, extending Wintner's result to Riemannian manifolds
and to second order ODE flows.

%Perhaps for this reason Wintner's theorem is not mentioned in the
%next cited work which consider the curved manifold case. The
%autonomous case for force fields $f$ of gradient type was considered
%in \cite{ebin70,weinstein70,gordon70}, while the non autonomous case
%for slightly more general fields, with an additional affine
%dependence on the velocity, has been only recently considered in
%\cite{candela12}.

Actually, we shall  also deal with an interesting related problem.
Let $(Q,a_t)$ be a manifold endowed with a time-dependent Riemannian
metric which is complete for every $t$. We are going to study the
completeness of the maximal solutions to the equation
\begin{equation}\label{hud}
\frac{D^{(t)} }{\dd t}\, \dot{q}=f(t,q(t), \dot{q}(t)),
\end{equation}
where $D^{(t)}$ is the affine connection of $(Q,a_t)$ and $f$ is as
above. Apparently there is no difference between the two equations
because the latter can be rewritten $\frac{D^{(0)} }{\dd t}\,
\dot{q}=f(t,q(t), \dot{q}(t))+T(\dot{q},\dot{q})$ where
$T=D^{(0)}-D^{(t)}$ is a tensor field of type $(1,2)$ symmetric in
the lower indices. However, we will be able to assure completeness
in the case in which
 $f$ has a component which depends quadratically in the
velocity only if the dependence is of the kind
%In general, in order to attain completeness we will not be able to
%put sufficient conditions on a force field $f$ which depends
%quadratically on the velocity saved for the terms of this type that
%are
naturally embodied in Eq.\ (\ref{hud}). This fact singles Eq.\
(\ref{hud}) as particularly well shaped for our completeness study.
Furthermore, and more importantly, the Euler-Lagrange equations of
classical mechanics that are met in rheonomic mechanical systems are
naturally written in this form where $a_t$ is the time dependent
matrix appearing in the Lagrangian kinetic term
\cite{greenwood77,goldstein65} (see also section \ref{app} below).

Our proof applies to the general case of Eq.\ (\ref{hud}) and is
based on a result in Riemannian geometry which allows us to avoid
the Nash embedding of $Q$ in $\mathbb{R}^n$ used in Gordon's proof
\cite{gordon70}. In Gordon's work this embedding was used in order
to construct a smooth proper function on the manifold. Actually, we
do not need smoothness and so we can take the distance squared as
proper function so as to give a straightforward geometrical meaning
to our bounds.

%improvement of Gordon's proof \cite{gordon70} which allows us to
%avoid the Nash embedding in $\mathbb{R}^n$ and at the same time

 %We shall make essential use of the Gronwall's inequality (actually
%we shall use the non-linear generalization by LaSalle and Bihari
%\cite{lasalle49,bihari56}).

% and we shall take advantage of the
%following  result in Riemannian geometry.
\section{Some distance inequalities}

In this section we obtain some inequalities for the Riemannian
distance.  They will allow us to construct proper functions which
are not necessarily smooth.

\begin{proposition} \label{buy}
Let $(Q,a_t)$ be a  Riemannian manifold of class $C^3$, and suppose
that for each $t$ the metric $a_t$ is $C^2$ with respect to the time
and space coordinates. Let $\rho_t\colon Q\times Q \to [0, +\infty]$ be
the Riemannian distance of $(Q,a_t)$, then $\rho_t(p,q)$ is
continuous in $(t,p,q)$.

Moreover, for fixed $p\in Q$, defined $R: \mathbb{R}\times Q \to
[1,+\infty)$ and $E\colon \mathbb{R}\times TQ \to [1,+\infty)$ with (here
$\Vert v\Vert_t^2=a_t(v,v)$)
\begin{align}
R(t,q)&= 1+\rho_t(p,q),\\
E(t,q,v)&=1+\rho_t(p,q)^2+\Vert v\Vert_t^2, \label{nux}
\end{align}
the functions $\vert t\vert+ R(t,q)$ and $\vert t\vert+ E(t,q,v)$,
respectively over $\mathbb{R}\times Q$ and $\mathbb{R}\times TQ$,
are proper.
\end{proposition}

\begin{proof}
Let us prove continuity at $(t,p,q)$. Let $\epsilon>0$, using the
$\sigma$-compactness of $Q$ we can find a metric $\hat{a}$ which is
at every point larger than $a_t$, for every $t\in
[t-\epsilon,t+\epsilon]$ (in the sense that the unit balls of
$\hat{a}$ on the tangent space are contained in those of $a_t$) (the
metric $\hat{a}$ need not be complete). The distance $\hat\rho$ is
continuous and the topology induced by $\hat\rho$ coincides with the
manifold topology. Using the triangle inequality for $t'\in
[t-\epsilon,t+\epsilon]$ we obtain
\begin{align*}
\vert \rho_{t'}(p',q')-\rho_t(p,q)\vert & \le \vert
\rho_{t'}(p',q')-\rho_{t'}(p,q) \vert+\vert
\rho_{t'}(p,q)-\rho_t(p,q)\vert \\
&\le \vert \rho_{t'}(p',q')-\rho_{t'}(p,q') \vert+\vert
\rho_{t'}(p,q')
-\rho_{t'}(p,q) \vert\\ & \quad +\vert \rho_{t'}(p,q)-\rho_t(p,q)\vert\\
& \le \rho_{t'}(p,p')+\rho_{t'}(q,q') +\vert
\rho_{t'}(p,q)-\rho_t(p,q)\vert \\
& \le \hat{\rho}(p,p')+\hat{\rho}(q,q') +\vert
\rho_{t'}(p,q)-\rho_t(p,q)\vert ,
\end{align*}
which implies that we need only to prove that for fixed $p,q\in Q$,
$\rho_{t'}(p,q) \to \rho_{t}(p,q)$ if $t'\to t$.

Let $\sigma\colon [0,1]\to Q$, $s\to \sigma(s)$, be a minimizing geodesic
which connects $p$ to $q$ in $(Q,a_t)$. Let $v$ be any vector field
over $\sigma([0,1])$. The function $a_{t'}(q')(v,v)$ is continuous
in $(t',q')$ and hence uniformly continuous over the compact set
$[t-\epsilon,t+\epsilon] \times \sigma([0,1])$, a fact which implies
the inequality for every $(t',q'), (t'',q'')\in
[t-\epsilon,t+\epsilon] \times \sigma([0,1])$
\[
\vert a_{t''}(q'')(v,v)-a_{t'}(q')(v,v)\vert \le
o(\hat{\rho}(q',q'')+\vert t''-t'\vert) \, .
\]
With $q''=q'$ the inequality states that for $t'\to t$,
$a_{t'}(v,v)\to a_t(v,v)$ uniformly over $\sigma([0,1])$. Let $v=\dd
\sigma/\dd s$ and let $l_t$ be the length functional for $(Q,a_t)$,
then $l_{t'}(\sigma)\to l_t(\sigma)=\rho_t(p,q)$, for $t'\to t$,
which implies that for every $\delta>0$ we have for $t'$
sufficiently close to $t$, $\rho_{t'}(p,q)\le l_{t'}(\sigma) \le
\rho_t(p,q)+\delta$. As $\rho_{t'}(p,q)$ stays bounded  by $
\rho_t(p,q)+\delta$ in the limit, if by contradiction,
$\rho_{t'}(p,q) \to \rho_{t}(p,q)$ does not hold, then we can find a
sequence $t_n\to t$ such that $\rho_{t_n}(p,q) \to L\ne
\rho_{t}(p,q)$, $L< \rho_t(p,q)+\delta$. As $\delta$ is arbitrary,
$L\le \rho_t(p,q)$, and hence $L< \rho_t(p,q)$.

Let $\exp^t$ be the exponential map for $(Q,a_t)$ and let $v_n\in
T_pQ$ be such that $\exp^{t_n}_p v_n=q$, thus $\Vert v_n
\Vert_{t_n}=\rho_{t_n}(p,q)$. Since $a_{t_n}\vert_p\to a_t\vert_p$
the sequence $v_n$ converges (pass to a subsequence if necessary) to
some vector $v\in T_pQ$, and $\Vert v_n \Vert_{t_n}\to \Vert v
\Vert_{t}$, thus $\Vert v \Vert_{t}=L< \rho_t(p,q)$. The geodesic
equation is a first order differential equation over $TQ$, thus
standard results \cite[Theor.\ 3.1]{hartman64} on the continuity of
first order (on $TQ$ in this case) differential equations with
respect to initial conditions ($v$ in this case) and external
parameters ($t$ in this case) imply that $q=\exp_p^{t_n} v_n \to
\exp_p^t v$. Since $\Vert v \Vert_{t}< \rho_t(p,q)$ we have that
$\exp_p^t v\ne q$. The contradiction proves that $\rho_{t'}(p,q)\to
\rho_{t}(p,q)$.

Let us prove that $F:=\vert t\vert+ E(t,q,v)$ is proper (the proof
that $\vert t\vert+R(t,q)$ is proper is contained in this one).
Clearly $F$ is continuous thus the inverse image of a compact set is
closed. If there is a compact set $K$ such $F^{-1}(K)$ is not
compact, then we can assume with no loss of generality $K=[-B,B]$
for some $B>0$, and we can find a sequence $(t_n,q_n,v_n)$ which
escapes every compact set of $\mathbb{R}\times TQ$ and is such that
 $\vert F(t_n,q_n,v_n)\vert\le B$. However, due to the
expression of $F$, $\vert t_n\vert \le B$, thus we can assume with
no loss of generality that $t_n\to t$ for some $t\in [-B,B]$.
Moreover, $\rho_{t_n}(p,q_n)^2\le B$ thus let $w_n \in T_pQ$ be a
vector such that $\exp^{t_n}_p w_n=q_n$; we have $\Vert
w_n\Vert_{t_n}=\rho_{t_n}(p,q_n)\le B^{1/2}$. Since
$a_{t_n}\vert_p\to a_t\vert_p$ we can assume with no loss of
generality (i.e.\ passing to a subsequence if necessary) that
$w_n\to w\in T_pQ$. Using again \cite[Theor.\ 3.1]{hartman64} we
obtain $q_n=\exp^{t_n}_p w_n \to \exp^t_p w=:\hat q$. As a
consequence, the sequence $q_n$ is contained in a compact set
$\tilde{K}\ni \hat q$. We can find a metric $\check{a}$ which is
smaller than $a_t$ in $\tilde{K}$ for $t\in [-B,B]$ (in the sense
clarified above). Thus the bound on $F(t_n,q_n,v_n)$ implies
$\check{a}(v_n,v_n)\le B$ which proves that the sequence
$(t_n,q_n,v_n)$ is actually contained in a compact set, a
contradiction.
\end{proof}

We need a simple preliminary lemma (for a more general version see \cite[Lemma 16.4]{amann90}).

\begin{lemma} \label{sop}
Let $f\colon [a,b]\to \mathbb{R}$ be a continuous function whose right
upper Dini derivative satisfies
\[
D^+ f(x):=\limsup_{h\to 0^+} \frac{f(x+h)-f(x)}{h} \le g'(x),
\]
 where $g\colon [a,b] \to \mathbb{R}$ is a $C^1$
function, then $f-g$ is non-increasing over the interval $[a,b]$.
\end{lemma}

\begin{proof}
Let us define $F:=f-g$, so that $D^+F=D^+f-g' \le 0$. Suppose by
contradiction, that there are $a', b'\in [a,b]$, $a'<b'$, such that
$F(b')>F(a')$,  and let $r=\frac{F(b')-F(a')}{b'-a'}>0$. The
continuous function $h=F-\frac{r}{2}(x-a')$ has a minimum at $c\in
[a',b']$, and since $h(b')=\frac{F(a')+F(b')}{2}>F(a')=h(a')$, we
have $c\ne b'$. Thus $D^+F(c)=D^+h(c)+\frac{r}{2}\ge \frac{r}{2}>0$,
a contradiction.
\end{proof}

The proof of the next proposition would be considerably shortened
under the assumption $\partial_t a_t=0$. However, we shall need the
following version in order to deal with a time dependent
metric.\footnote{An heuristic way of obtaining Eq.\ (\ref{ers})
consists in working out  $\partial_t[\rho_t^2(p,q(t))]$, expressing
the squared distance as an 'energy' action integral over a minimal
geodesic. The reader has to use the Cauchy-Schwarz inequality for
the scalar product that there appears and then integrate. This
method is not rigorous since even for $\partial_t a_t=0$,
$\rho^2(p,q(t))$ is not always differentiable when $q(t)$ passes
through a cut point to $p$. This problem cannot be easily amended
since the cut points are not necessarily isolated.}

\begin{proposition}
Let $(Q,a_t)$  and $\rho_t$ be as in Prop.\ \ref{buy}.
\begin{itemize}
\item[(i)]
Suppose that the time derivative of the metric $a_t$ is bounded as
follows: there is a continuous function  $g\colon [1,+\infty)\to
[1,+\infty)$ such that
\begin{equation} \label{biq}
\pm  (\partial_t a_t)(v,v) \le 2 g(1+\rho_t(p,q)) \, a_t(v,v),
\end{equation}
at every point $(t,q,v)\in \mathbb{R}\times TQ$. Let $q\colon I\to Q$ be
a $C^1$ curve, then  for $\underline{t}, \overline{t}\in I$,
$\underline{t}<\overline{t}$, we have respectively
\begin{align}
\pm [\rho_{\overline{t}}(p, q(\overline{t}))-\!
\rho_{\underline{t}}(p, q(\underline{t})) ] &\le \!\!
\int_{\underline{t}}^{\overline{t}}\! \{ \Vert \dot{q} \Vert_t
+g(1\!+\!\rho_t(p,q(t)))\, \rho_t(p, q(t))\} \,\dd t, \label{erx}
\end{align}
\item[(ii)] Suppose  that the time derivative of the metric $a_t$ is
bounded as follows: there is a continuous function $g\colon [1,+\infty)\to
[1,+\infty)$ such that
\begin{equation} \label{biq2}
\pm  (\partial_t a_t)(v,v) \le 2 g(1+\rho_t(p,q)^2) \, a_t(v,v),
\end{equation}
at every point $(t,q,v)\in \mathbb{R}\times TQ$. Let $q\colon I\to Q$ be
a $C^1$ curve, then  for $\underline{t}, \overline{t}\in I$,
$\underline{t}<\overline{t}$, we have respectively
\begin{align}
\pm [\rho_{\overline{t}}(p,
q(\overline{t}))^2-&\rho_{\underline{t}}(p,
q(\underline{t}))^2 ]\le \nonumber \\
& \!\!\!\!\!\!\!\!\!\!\!\!\!\!\!\!\!\!\! \le  2\!\!
\int_{\underline{t}}^{\overline{t}}\! \{\rho_t(p, q(t)) \Vert
\dot{q} \Vert_t +g(1+\rho_t(p,q(t))^2)\, \rho_t(p, q(t))^2\} \,\dd
t. \label{ers}
\end{align}
\end{itemize}

\end{proposition}

\begin{proof}

Let us prove the inequality Eqs.\ (\ref{erx}) and (\ref{ers}) with
the left-hand side replaced respectively by $\rho_t(p,
q(\overline{t}))-\rho_t(p, q(\underline{t}))$ and $\rho_t(p,
q(\overline{t}))^2-\rho_t(p, q(\underline{t}))^2$. The other
direction follows considering the curve $q'(t)=q(-t)$, $q'\colon
[-\overline{t}, -\underline{t}] \to Q$, and the Riemannian spaces
$(Q,a'_{t})$, $a'_{t}=a_{-t}$, in such a way that $\rho_t'(p,
q'(t))=\rho_{-t}(p,q(-t))$.

For any positive integer $n$, let
$\epsilon=(\overline{t}-\underline{t})/n$, and let
$t_k=\underline{t}+k \epsilon$. Let us abbreviate $\rho_{t_k}$ with
$\rho_k$ and $q(t_k)$ with $q_k$. Then for $k=0,1,\ldots, n-1$, we
have, by the triangle inequality
\begin{align*}
\rho_{k+1}(p, q_{k+1})-\!\rho_k(p, q_k) &\le \vert \rho_{k+1}(p,
q_{k+\!1})\!-\!\rho_{k+\!1}(p,
q_{k})\vert +[ \rho_{k+\!1}(p, q_{k})\!-\!\rho_k(p, q_k)]\\
&\le A_k+B_k ,\\
 \rho_{k+1}(p,
q_{k+1})^2-\!\rho_k(p, q_k)^2 &=[ \rho_{k+1}(p, q_{k+1})-\!\rho_k(p,
q_k)]\,  [\rho_{k+1}(p, q_{k+1})+\!\rho_k(p,
q_{k})]\\
&\le \{\vert \rho_{k+1}(p, q_{k+\!1})\!-\!\rho_{k+\!1}(p,
q_{k})\vert +[ \rho_{k+\!1}(p, q_{k})\!-\!\rho_k(p, q_k)]\}\, \\&
\quad \ [\rho_{k+1}(p,
q_{k+1})+\!\rho_k(p, q_{k})]\\
&\le 2[A_k+B_k]C_k ,
\end{align*}
where
\begin{align*}
A_k&= \rho_{k+1}(q_k, q_{k+1}) ,\\
B_k&= \rho_{k+1}(p, q_{k})-\!\rho_k(p, q_k) ,\\
C_k&=  [\rho_{k+1}(p, q_{k+1})+\!\rho_k(p, q_{k})]/2.
\end{align*}
Summing over $k$ and taking the limit for $n\to+\infty$
\begin{align*}
\rho_{\overline{t}}(p, q(\overline{t}))- \rho_{\underline{t}}(p,
q(\underline{t})) &\le \lim_{n\to
+\infty}\sum_{k=0}A_k+\lim_{n\to +\infty}\sum_{k=0} B_k ,
\\
\rho_{\overline{t}}(p, q(\overline{t}))^2- \rho_{\underline{t}}(p,
q(\underline{t}))^2&\le 2\lim_{n\to
+\infty}\sum_{k=0}A_kC_k+2\lim_{n\to +2\infty}\sum_{k=0} B_k C_k.
\end{align*}
By the continuity of $\rho_t(p,q(t))$, there is a point
$\tilde{t}_k\in [t_k, t_{k+1}]$ such that
\[
\rho_{\tilde{t}_k}(p,q(\tilde{t}_k))=C_k.
\]
Moreover, there is some $\hat{t}_k\in [t_k, t_{k+1}]$ such that
\[A_k=\rho_{k+1}(q_k, q_{k+1}) \le \int_{t_k}^{t_{k+1}}
\Vert \dot{q}\Vert_{t_{k+1}} \,\dd t = \Vert
\dot{q}(\hat{t}_k)\Vert_{t_{k+1}} \, \epsilon,
\]
thus
\[
2A_k C_k\le 2 \rho_{\tilde{t}_k}(p,q(\tilde{t}_k)) \Vert
\dot{q}(\hat{t}_k)\Vert_{t_{k+1}} \, \epsilon,
\]
and
\begin{align*}
\lim_{n\to +\infty}\sum^{n-1}_{k=0} A_k&\le  \lim_{n\to +\infty}
\sum_{k=0}^{n-1}
\Vert \dot{q}(\hat{t}_k)\Vert_{t_{k+1}}\, \frac{1}{n},\\
2\lim_{n\to +\infty}\sum^{n-1}_{k=0} A_k C_k&\le  \lim_{n\to
+\infty} \sum_{k=0}^{n-1}
 2 \,\rho_{\tilde{t}_k}(p,q(\tilde{t}_k))\,
\Vert \dot{q}(\hat{t}_k)\Vert_{t_{k+1}}\, \frac{1}{n}.
\end{align*}
We obtain the first term in the integral argument in the right-hand
side of Eqs.\ (\ref{erx})-(\ref{ers}). Indeed, the right-hand side
of the above equation is the Riemann integral of a continuous
function \cite{hobson21}  (observe that
$\sqrt{a_{t'}(\dot{q}(t),\dot{q}(t))}$ regarded as a function of
$(t',t)$ is continuous and hence uniformly continuous over the
compact set $[\underline{t},\overline{t}]^2$, thus in the previous
expression $(t_{k+1},\hat{t}_k)$ can be replaced by
$(\tilde{t}_k,\tilde{t}_k)$ with a total error which can be made
arbitrarily small).

%
%\begin{align*}
%\rho (p, q(t_{k+1}))^2-\rho(p, q(t_k))^2 &\le 2 \,
%\rho(p,q(\tilde{t}_k))\,
% \rho(q(t_k),
%q(t_{k+1})) \\&\le 2 \,\rho(p,q(\tilde{t}_k))
%\sqrt{a(\dot{q},\dot{q})}(\hat{t}_k) \, \epsilon.
%\end{align*}
%Summing over $k$, $\rho(p, q(\overline{t}))^2-\rho(p,
%q(\underline{t}))^2\le  2\sum_{k=0}^{n-1}\frac{1}{n }\,
%\rho(p,q(\tilde{t}_k)) \sqrt{a(\dot{q},\dot{q})}(\hat{t}_k)$, and
%taking the limit $\epsilon \to 0$ we obtain Eq.\ (\ref{ers}), where
%the right-hand side is the Riemann integral of a continuous function
%\cite{hobson21} (observe that $\sqrt{a(\dot{q},\dot{q})}(t)$ is
%continuous and hence uniformly continuous over
%$[\underline{t},\overline{t}]$, thus in the previous expression
%$\hat{t}_k$ can be replaced by $\tilde{t}_k$ with an error which
%vanishes in the limit).

As for the remaining term, let $l_t$ be the length functional of
$(Q,a_t)$. For any fixed $C^1$ curve $\eta(s)$ the function
$l_t[\eta]$ is $C^1$ in $t$ because (here $\dd \eta/\dd s$ is
denoted $\eta'$)
\[
\partial_t l_t[\eta]=\int_\eta \partial_t\sqrt{a_t(\eta',\eta')}
\, \dd s=\int_\eta \frac{1}{2\sqrt{a_t(\eta',\eta')}}\, {(\partial_t
a_t)(\eta',\eta')} \, \dd s.
\]

We wish to bound $D^+_t \rho_{t}(p,r)$ at any time $t$ where $r$ is
arbitrary and does not depend on $t$.

 Let $\gamma_t$ be a minimizing geodesic of
$(Q,a_t)$ which connects  $p$ to $r$. We have
\[
\limsup_{\varepsilon\to 0^+}\frac{1}{\varepsilon}[
\rho_{t+\varepsilon}(p,r)-\rho_{t}(p,r)]\le \limsup_{\varepsilon\to
0^+}
\frac{1}{\varepsilon}\{l_{t+\varepsilon}[\gamma]-l_t[\gamma]\}=\partial_t
l_t[\gamma_t],
\]
where we have used the fact that at time $t+\epsilon$, for
$\epsilon>0$, $\gamma_t(s)$ is not necessarily minimizing. Thus,
using Eq.\ (\ref{biq}) or (\ref{biq2}) , we obtain (in cases (i) and
(ii) respectively)
\begin{align*}
D^+_t \rho_{t}(p,r)&\le \partial_t l_t[\gamma]=\int_{\gamma_t}
\frac{1}{2\sqrt{a_t(\gamma'_t,\gamma'_t)}}\, {(\partial_t
a_t)(\gamma'_t,\gamma'_t)} \, \dd s\\
&\le\int_{\gamma_t} g(1+\rho_t(p,\gamma_t(s))) \,
\sqrt{a_t(\gamma'_t,\gamma'_t)} \, \dd s.\\
 D^+_t \rho_{t}(p,r)&\le
\partial_t l_t[\gamma]=\int_{\gamma_t}
\frac{1}{2\sqrt{a_t(\gamma'_t,\gamma'_t)}}\, {(\partial_t
a_t)(\gamma'_t,\gamma'_t)} \, \dd s\\
&\le\int_{\gamma_t} g(1+\rho_t(p,\gamma_t(s))^2) \,
\sqrt{a_t(\gamma'_t,\gamma'_t)} \, \dd s.
\end{align*}
But as at time $t$ the geodesic $\gamma_t$ is minimizing and starts
from $p$, the function $\rho_t(p,\gamma_t(s))$ grows with $s$, and
hence  (in cases (i) and (ii) respectively)
\begin{align*}
D^+_t \rho_{t}(p,r)&\le g(1+\rho_t(p,r)) \int_{\gamma_t} \,
\sqrt{a_t(\gamma'_t,\gamma'_t)} \, \dd s=g(1+\rho_t(p,r))
\rho_t(p,r),\\
 D^+_t \rho_{t}(p,r)&\le g(1+\rho_t(p,r)^2)
\int_{\gamma_t} \, \sqrt{a_t(\gamma'_t,\gamma'_t)} \, \dd
s=g(1+\rho_t(p,r)^2) \rho_t(p,r).
\end{align*}
By  Lemma \ref{sop} we have for $\hat{t}\le \check{t}$  (in cases
(i) and (ii) respectively)
\begin{align}
\rho_{\check{t}}(p,r)-\rho_{\hat{t}}(p,r)&\le
\int_{\hat{t}}^{\check{t}} g(1+\rho_t(p,r)) \rho_t(p,r) \,\dd t,\\
\rho_{\check{t}}(p,r)-\rho_{\hat{t}}(p,r)&\le
\int_{\hat{t}}^{\check{t}} g(1+\rho_t(p,r)^2) \rho_t(p,r) \,\dd t.
\end{align}
Choosing $\hat{t}=t_{k}$, $\check{t}=t_{k+1}$, $r=q_k$ we obtain (in
cases (i) and (ii) respectively)
\begin{align*}
B_k &\le \int_{t_k}^{t_{k+1}} g(1+\rho_t(p,q_k)) \rho_t(p,q_k) \dd
t \le g(1+\rho_{t'}(p,q_k)) \rho_{t'}(p,q_k) \, \frac{1}{n} ,\\
B_k &\le \int_{t_k}^{t_{k+1}} g(1+\rho_t(p,q_k)^2) \rho_t(p,q_k) \dd
t \le g(1+\rho_{t'}(p,q_k)^2) \rho_{t'}(p,q_k) \, \frac{1}{n} ,
\end{align*}
for some (different) $t'\in [t_k,t_{k+1}]$. Finally, (in cases (i) and (ii)
respectively)
\begin{align*}
\sum_{k=0}^{n-1} B_k &\le \sum_{k=0}^{n-1} g(1+\rho_{t'}(p,q_k)) \rho_{t'}(p,q_k) \, \frac{1}{n} ,  \\
2\sum_{k=0}^{n-1} B_kC_k &\le 2 \sum_{k=0}^{n-1}
\rho_{\tilde{t}_k}(p,q(\tilde{t}_k)) g(1+\rho_{t'}(p,q_k)^2)
\rho_{t'}(p,q_k) \, \frac{1}{n} .
\end{align*}
Using the continuity and hence uniform continuity of
$g(1+\rho_{t'}(p,q_k)) \rho_{t'}(p,q_k) $ as a function of $t$ on
the compact set $[\underline{t},\overline{t}]$, or of
$\rho_{t}(p,q(t)) g(1+\rho_{t'}(p,q_k)^2) \rho_{t'}(p,q_k)$ as a
function of $(t,t')$ on the compact set
$[\underline{t},\overline{t}]^2$ we conclude that in the limit $n\to
\infty$ the right-hand sides of the previous inequalities converge
respectively to the Riemann integrals 
\[
\int_{\underline{t}}^{\overline{t}} [g(1+\rho_t(p,q(t)))\, \rho_t(p,
q(t))] \,\dd t \quad \textrm{and} \quad 2 \int_{\underline{t}}^{\overline{t}}
[g(1+\rho_t(p,q(t))^2)\, \rho_t(p, q(t))^2] \,\dd t.
\]
\end{proof}
%
%By the lemma for $t'>t$, $\rho_{t'}(r_1,r_2)- l_{t'}[\gamma]\le 0$
%whenever $\gamma$ minimizes the
%
%
%
%
%%\[
%%D^+_t \rho_{t}(r_1,r_2)^2\le 2 \rho_{t}(r_1,r_2)\partial_t
%%l_t[\gamma].
%%\]
%Let us consider the special case in which $r_1=p$, $r_2=q(t')$ and
%$\gamma_{t'}$ is a minimizing geodesic for $(Q,a_{t'})$ connecting
%the two points. We have $(D^+_t \rho_{t})(t')(p,q(t'))\le
%(\partial_t l_t)(t')[\gamma_{t'}]$ and
%$\rho_{t'}(p,q(t'))=l_{t'}[\gamma_{t'}]$. Thus by the lemma we have
%for $t''\ge t'$
%\[
%\rho_{t''}(p,q(t''))\le l_{t''}[\gamma_{t''}]\le
%\rho_{{t''}}(p,q({t'}))+\int_{t'}^{t''} \partial_s l_s[\gamma_{t'}]
%\dd s .
%\]
%
%
%Let $\bar{t}\in\mathbb{R}$ and suppose that  $r_1=p$,
%$r_2=q(\bar{t})$ and $\gamma$ is a minimizing geodesic for
%$(Q,a_{\bar t})$ which starts from $p$ and ends at $q(\bar{t})$.
%Observe that for $\rho_{\bar{t}}(p,q(\bar{t}))=
%l_t[\gamma](\bar{t})$, thus by the lemma, for $t'\ge \bar{t}$ we
%have
%\[
%\rho_{t'}(p,q(t'))\le l_{t'}[\gamma]\le
%\rho_{\bar{t}}(p,q(\bar{t}))+\int_{\bar{t}}^t \partial_s l_s[\gamma]
%\dd s .
%\]
%
%
%
%
%
%Using Eq.\ (\ref{biq}) and the fact that $\rho_t(p,\gamma(s))$ is
%increasing by the minimality of $\gamma$, we obtain (recall  that
%$l_t[\gamma]=\rho_t(p,q(t))$)
%\[
%\partial_t l_t[\gamma]\le \int_\eta g(1+\rho_t(p,\gamma(s))^2) \,
%\sqrt{a_t(\eta_s,\eta_s)} \, \dd s\le g(1+\rho_t(p,q(t))^2)
%\rho_t(p,q(t)).
%\]

\subsection{Generalization of Wintner's theorem to Riemannian
manifolds}

In this section we generalize Wintner's theorem to Riemannian spaces
with possibly time-dependent metrics.

Let us recall LaSalle-Bihari \cite{lasalle49,bihari56}
generalization of Gronwall's inequality. We give here a kind of
two-sided generalization.

\begin{theorem} \label{nuf}
Let $x\colon  I \to [0,+\infty)$ be a continuous function which satisfies
the inequality:
%\begin{equation} \label{xra}
%x (t) \le x(a) + \beta \vert \int_a^t
%    \Psi(s)\, \omega(x (s))\, \dd s\vert, \qquad a,t\in I
%\end{equation}
%or equivalently
\begin{equation} \label{xrb}
\pm [ x (t') -x(t)] \le \beta  \int_t^{t'}
    \Psi(s) \,\omega(x (s))\, \dd s, \qquad t,t'\in I, \, t<t',
\end{equation}
where $\beta>0$, $\Psi\colon \mathbb{R} \to [0,+\infty)$, and    $\omega\colon [0,+\infty) \to (0, +\infty)$ is continuous and non-decreasing.
Then we have the respective estimates
\begin{equation} \label{xrc}
\pm [\Phi(x(t'))- \Phi(x(t))]\le  \beta  \int_t^{t'}
    \Psi(s)\, \dd s, \qquad t,t'\in I, \, t<t',
\end{equation}
where $\Phi\colon \mathbb{R} \to \mathbb{R}$  is given by
\[
\Phi(u) := \int_{u_0}^u \frac{\dd s}{\omega(s)} , \qquad u \in
\mathbb{R}.
\]
%Finally,  the inequality (\ref{xra}) holds  in plus case for the
%absolute value (i.e.\ for $t>a$) iff the inequality (\ref{xrb}) holds in the plus case,
%in which case (\ref{xrc}) holds in the plus case (i.e.\ for $t>a$). Analogously, there
%is a correspondence for the minus cases.
\end{theorem}

%\begin{remark}
%The inequality \ref{xra} holds only in plus case for the absolute
%value, iff the same holds for \ref{xrb} and \ref{xrc}, and hence
%analogously in the minus case. W
%\end{remark}

\begin{proof}
The minus version can be obtained from the usual plus version by
considering the function $\tilde{x}(t)=x(-t)$, and applying the plus
version to $\tilde{x}$.
\end{proof}

%
%\begin{theorem}
%Let $x : [a, b] \to [0,+\infty)$ be a continuous function which
%satisfies the inequality:
%
%\[
%x (t) \le x(a) + \beta \int_a^t
%    \omega(x (s))\, \dd s, \qquad t\in [a, b]
%\]
%where    $\Psi: [a, b] \to [0,+\infty)$ is continuous and $\omega :
%[0,+\infty) \to (0, +\infty)$ is continuous and monotone-increasing.
%Then the estimation
%\[
%x (t)\le \Phi^{-1} \left( \Phi(M) + \beta (t-a)
% \right), \qquad t \in [a, b]
%\]
%holds, where $\Phi: \mathbb{R} \to \mathbb{R}$  is given by
%\[
%\Phi(u) := \int_{u_0}^u \frac{\dd s}{\omega(s)} , \qquad u \in
%\mathbb{R}.
%\]
%\end{theorem}

\begin{remark} \label{rem}
With $g$ or $g_r$ we denote a non-decreasing $C^0$ function
$g_r:[1,+\infty) \to [1,+\infty)$ with the property that the
increasing function $G_r\colon [1,+\infty) \to [0,+\infty)$
\[
G_r(y):=\int_1^y \frac{1}{x g_r(x)}\, \dd x ,
\]
diverges for $y\to +\infty$. The typical choice will be
$g=cnst.\ge1$, but there are choices that strengthen the next
theorems, e.g.\ $g=\ln(\eta+ x)$ or $g=\ln(\eta+ x) \ln(\eta+
\ln(\eta+ x))$ and so on,  where $\eta=e-1$.
\end{remark}

On first reading one can just consider the simple case $g_r=cnst.\ge
1$, $\partial_t a_t=0$. We are ready to  generalize  Wintner's
theorem to Riemannian manifolds.

\begin{theorem}
Let $(Q,a_t)$ be a 1-parameter family of complete Riemannian
manifolds as in Prop.\ \ref{buy}, and let $f:\mathbb{R}\times Q \to
TQ$ be a $C^0$ field. Let $p\in Q$ and let $R:  \mathbb{R}\times Q
\to [1,+\infty)$ be given  by
\[
R(t,q)=1+\rho_t(p,q).
\]
%\begin{itemize}
%\item[(a)]
Suppose that for every compact interval $[-r,r] \subset \mathbb{R}$
there is a function $g_r$ as in remark \ref{rem} such that for every
$(t,q)\in [-r,r]\times Q$
\begin{align}
\pm (\partial_t a_t)(v,v) &\le 2 g_r(R(t,q)) \,
a_t(v,v),\\
\Vert \nu(t,q)\Vert_t & \le g_r(R(t,q))\, {R(t,q)},
\end{align}
then the maximal solutions to the first order equation
\[
 \dot{q}=\nu(t,q(t)),
\]
are complete in the forward (resp.\ backward) direction.
%\item[(b)] Suppose that  there is a function $g$ as above  and a function $\Psi: \mathbb{R}\to (0,+\infty)$ such that for every
%$(t,q)\in [t_0,\infty)\times Q$
%\begin{align}
% (\partial_t a_t)(v,v) &\le 2  \Psi(t) \, g(R(t,q)) \,
%a_t(v,v),\\
%\Vert \nu(t,q)\Vert_t & \le \Psi(t) \, g(R(t,q))\, {R(t,q)},
%\end{align}
%and $\int_{t_0}^{+\infty} \Psi(s)\dd s<+\infty$, then the maximal
%solutions to the first order equation
%\[
% \dot{q}=\nu(t,q(t)),
%\]
%are complete in the forward direction and converge to some
%$q_{\infty}$ for $t\to +\infty$.
%\end{itemize}
\end{theorem}

\begin{remark}
Since we do not need the uniqueness of the solution, we just ask
$\nu$ to be continuous rather than Lipschitz \cite{hale80}. The
existence of some maximal solutions is assured by \cite[Theor.\
3.1]{hartman64}. The theorem can be easily generalized to
differential inclusions, i.e.\ to the case in which $\nu$ is a lower
semi-continuous set valued mapping and $\nu(t,q)$ is a convex set of
$T_qQ$ for each $(t,q)$. In this case the existence of some maximal
solutions is assured by \cite[Theor.\ 1, Chap.\ 2]{aubin84}.
\end{remark}

\begin{proof}
%Proof of (a).
Let the initial condition be ${q}(t_0)=q_0$, and suppose by
contradiction that $q(t)$ is a maximal solution whose interval of
definition $I$ is bounded on the right, i.e.\ $I\subset (-\infty,B]$
(resp.\ on the left $I \subset [-B,+\infty)$) for some $B>0$. Let us
observe that
\[
\Vert \dot{q}\Vert_t=\Vert \nu\Vert_t\le g_B(R(t,q(t)))\,
{R(t,q(t))},
\]
thus plugged in Eq.\ (\ref{erx}) with $g=g_B$, we obtain for every
$t_1,t_2\in I$, $t_1\le t_2$,
\begin{equation}
\pm [R({t_2}, q({t_2}))- R({t_1}, q({t_1}))]\le 2\!\!
\int_{{t_1}}^{{t_2}}\! g_B(R(s,q(s)))\, {R(s,q(s))} \,\dd s.
\end{equation}
%We are going to prove that if $t\in I$, $t> t_0$,
%\begin{equation} \label{kct}
%R(t,q(t))\le \alpha+\beta\int_{t_0}^t g_B(R(s,q(s)))  \, R(s,q(s))
%\,\dd s.
%\end{equation}
%for  constants $\alpha=R(t_0,q(t_0))$, $\beta>0$.
Thus,   by LaSalle-Bihari generalization of Gronwall's inequality
(Theor.\ \ref{nuf}) we have for $t\in I$, $t>t_0$ (resp.\ $t<t_0$)
\begin{align*}
R(t,q(t)) &\le  G_B^{-1}(G_B(R(t_0,q(t_0)))+ 2 \vert t-t_0\vert) \\
&\le G_B^{-1}(G_B(R(t_0,q(t_0)))+ 4 B)<+\infty.
\end{align*}
As $R$ is proper and $\vert t\vert$ is bounded by $B$ in the forward (resp. backward) direction, $q(t)$ cannot
escape every compact set and hence it must be complete in that direction \cite[Theor.\
3.1, Chap.\ II]{hartman64}, a contradiction.
%Similarly
%we are going to prove that if $t\in I$, $t< t_0$,
%\begin{equation} \label{kct}
%R(t,q(t))\ge \alpha-\beta\int_{t}^{t_0} g_B(R(s,q(s)))  \, R(s,q(s))
%\,\dd s.
%\end{equation}
%for (the same) constants $\alpha,\beta>0$. A similar argument shows
%that  $q(t)$ cannot escape every compact set in the past direction
%and hence it must be complete in the past direction, a
%contradiction.
% In order to prove the inequality (\ref{kct}),
\end{proof}

\section{The completeness of maximal solutions to second order
ODEs}

The following notation will simplify the statement of the theorem.
In this section the functions $g_r$ and $G_r$ are defined as in the
previous section.
 Let $F_t\colon  Q \to \bigwedge^2 Q $ be a time dependent 2-form
field. By  $F_t^\sharp\colon TQ\to TQ$ we denote  the corresponding
endomorphism of the tangent space defined by
$F_t^\sharp(v)=a_t^{-1}(F_t(\cdot,v))$. On first reading one can just
consider the simple case $g_r=cnst.\ge 1$, $F_t=0$, $\partial_t
a_t=0$.
%With $g$ or $g_r$ we denote a non-decreasing $C^0$ function
%$g_r~:~[1,+\infty) \to [1,+\infty)$ with the property that the
%increasing function $G_r:[1,+\infty) \to [0,+\infty)$
%\[
%G_r(y):=\int_1^y \frac{1}{x g_r(x)}\, \dd x
%\]
%diverges for $y\to +\infty$. The typical choice will be
%$g=cnst.\ge1$, but there are choices that strengthen the next
%theorem, e.g.\ $g=\ln(\eta+ x)$ or $g=\ln(\eta+ x) \ln(\eta+
%\ln(\eta+ x))$ and so on,  where $\eta=e-1$. On first reading one
%can just consider the simple case $g_r=cnst.\ge 1$, $F_t^\sharp=0$,
%$\partial_t a_t=0$.
We are ready to state the theorem.

\begin{theorem} \label{bso}
Let $(Q,a_t)$ be a 1-parameter family of complete Riemannian
manifolds as in Prop.\ \ref{buy}, and let $f\colon \mathbb{R}\times TQ \to
TQ$ be a $C^0$ field. Let $p\in Q$ and let $E\colon   \mathbb{R}\times TQ
\to [1,+\infty)$ be given as above by
\[
E(t,q,v)=1+\Vert v \Vert_t^2+ \rho_t(p,q)^2.
\]
Suppose that for every compact interval $[-r,r] \subset \mathbb{R}$
there are a constant $K(r)>0$, a function $g_r$ as in remark
\ref{rem}, and a continuous 2-form field  $F_t$,
 such that for every $(t,q,v)\in [-r,r]\times TQ$
\begin{align}
\pm [ (\partial_t a_t)(v,v) ]&\le 2 g_r(1+\rho_t(p,q)^2) \,
a_t(v,v),\\
\Vert f(t,q, v)-F^\sharp_t(v) \Vert_t & \le
g_r(E(t,q,v))\,\frac{E(t,q,v)}{K+\Vert v \Vert_t}. \label{bis}
\end{align}
Then the maximal solutions to the second order equation
\[
\frac{D^{(t)} }{\dd t}\, \dot{q}=f(t,q(t), \dot{q}(t)),
\]
are complete in the forward  (resp.\ backward) direction.

\end{theorem}

\begin{remark}
The inequality (\ref{bis}) mentions  $F^\sharp_t$ in order to
clarify that the inclusion of force field components of
`electromagnetic' or `Coriolis' type can only enlarge  the domain of
the maximal solutions (see also \cite{candela13}). Actually, the
left-hand side of Eq.\ (\ref{bis})  could be replaced, just in the
study of forward completeness, by $\Vert f(t,q,
v)-F^\sharp_t(v)+h(t,q,v) \, v \Vert_t$ where $h: \mathbb{R}\times
TQ \to [0,+\infty)$. This can be easily understood from inspection
of Eq.\ (\ref{bkd}). This fact proves that the introduction of
force components which represent friction forces proportional to the
velocity can only enlarge the domain of the maximal solutions and
hence can only make it easier to attain forward completeness, see
also \cite{weinstein70}. However, these maximal solutions could be
incomplete in the backward direction.
\end{remark}

\begin{remark}
Clearly, the assumptions of the theorem are satisfied if $\partial_t
a_t=0$ and for each $s$ there are positive constants $K_0(r),
K_1(r), K_2(r)$, $K_3(r), K_4(r),$ such that for $(t,q,v)\in
[-r,r]\times TQ$
\[
\Vert f(t,q, v)\Vert \le K_0+K_1 \rho(p,q) +K_2 \Vert v \Vert + K_3
\frac{\rho(p,q)^2}{K_4 + \Vert v \Vert }.
\]
It suffices to choose  $g_r$ to be a sufficiently large constant and
$K_4=K$. Thus, in this case the second order ODE is complete. Of
course the asymptotic behavior of $f$ could  be faster than linear
for instance of the form $\sim \rho \ln(\eta+\rho)$, it suffices to
consider a non-trivial $g_r$. Observe that we did not impose any
type of dependence of $f$ on $v$, e.g.\ linear, as in previous
works, thus $f(v)$ could be quite general. Also we obtain a new type
of sufficient asymptotic bound, expressed by the last term of the
previous expression, which has been noticed here for the first time.
\end{remark}

\begin{remark}
As in the previous section, since we do not need the uniqueness of
the solution, we ask $f$ to be just  continuous rather than
Lipschitz. The existence of some maximal solutions is assured by
\cite[Theor.\ 3.1]{hartman64}. Theorem \ref{bso} can be easily
generalized to differential inclusions, i.e.\ to the case in which
$f$ is a lower semi-continuous set valued mapping and $f(t,q,v)$ is
a convex set of $TQ$ for each $(t,q,v)$. In this case the existence
of some maximal solutions is assured by \cite[Theor.\ 1, Chap.\
2]{aubin84}.
\end{remark}

\begin{proof}
%The function $E: \mathbb{R}\times TQ \to [1,+\infty)$, given by
%\[
%E(t,q,v)=1+a(v,v)+\rho(p,q)^2,
%\]
%is proper by the Hopf-Rinow theorem.
Let the initial condition be ${q}(t_0)=q_0$, $\dot{q}(t_0)=v_0$ and
suppose by contradiction that $q(t)$ is a maximal solution whose
interval of definition $I$ is bounded from above, i.e.\ $I\subset
(-\infty,B]$ for some $B>0$ (resp.\ from below, i.e.\ $I\subset [-B,
+\infty)$ for some $B>0$). Let us consider the curve for $t\in
\hat{I}=[-B,B]\cap I$; we can take $B$ sufficiently large so that
$t_0\in [-B,B]$. We are going to prove that if $t_1, t_2\in \hat I$,
$t_1<t_2$,
\begin{align}
\pm [ E(t_2,q(t_2),\dot{q}(t_2))-&E(t_1,q(t_1),\dot{q}(t_1))] \le \nonumber\\
& \le \beta\int_{t_1}^{t_2} g_B(E(s,q(s),\dot{q}(s)))  \,
E(s,q(s),\dot{q}(s)) \,\dd s. \label{kct}
\end{align}
for  a constant $\beta>0$. Thus, defined $E_0=
E(t_0,q(t_0),\dot{q}(t_0))$, by LaSalle-Bihari generalization of
Gronwall's inequality (Theor.\ \ref{nuf}) we have for $t\in \hat I$,
$t>t_0$ (resp.\ $t<t_0$)
\[
E(t,q(t),\dot{q}(t))\le  G_B^{-1}(G_B(E_0)+ \beta \vert t-t_0\vert)
\le G_B^{-1}(G_B(E_0)+ \beta 2B)<+\infty.
\]
As a consequence, as $F(t,q(t),\dot q(t))=\vert t\vert+E$ and $F$ is
proper by Prop.\ \ref{buy}, $(q(t),\dot{q})$ cannot escape every
compact set in the forward (resp.\ backward) direction and hence
this solution must be complete in that direction \cite[Theor.\ 3.1,
Chap. II]{hartman64}, a contradiction. In order to prove the
inequality (\ref{kct}), let us observe that
\begin{align}
\pm \frac{\dd }{\dd t}\, a_t(\dot q, \dot q)&= \pm 2  a_t(\dot q,
\frac{D^{(t)} \dot q}{\dd t})\pm (\partial_t a_t)(\dot q, \dot q)=
\pm 2 a_t(\dot q,
f)\pm (\partial_t a_t)(\dot q, \dot q)  \nonumber \\
&\le \pm 2 a_t(\dot q,
f-F^\sharp_t(\dot{q}))+2g_B(1+\rho_t(p,q(t))^2)
\Vert \dot{q} \Vert_t^2 \label{bkd} \\
&\le 2 \, \Vert \dot q\Vert_t \, \Vert f-F^\sharp_t(\dot{q}) \Vert_t
+2 g_B(1+\rho_t(p,q(t))^2) \Vert \dot{q} \Vert_t^2, \nonumber
\end{align}
which once integrated gives for $t_1,t_2\in I$, $t_1<t_2$,
\[
\pm [ a_{t_2}(\dot q, \dot q)(t_2)-a_{t_1}(\dot q, \dot q)(t_1)] \le
2 \int_{t_1}^{t_2} \{\Vert \dot q\Vert_s  \, \Vert f
-F^\sharp_s(\dot{q})\Vert_s+g_B(1+\rho_s(p,q(s))^2) \Vert \dot{q}
\Vert_s^2\} \, \dd s .
\]
Summing it to Eq. (\ref{ers}) where we make the choice $g= g_B$,
$\overline{t}=t_2$, $\underline{t}=t_1$, we obtain
\begin{align*}
\pm [ E(t_2,q(t_2),\dot{q}(t_2))-E(t_1,q(t_1),\dot{q}(t_1))] &\le 2
\int_{t_1}^{t_2} \Lambda(s) \, g_B(E) E\, \dd s,
\end{align*}
where we shortened the notation  introducing the function
\[
\Lambda(s)=\frac{ x(s) \Vert \dot q\Vert_s +g_B(1+x(s)^2)\, x(s)^2 +
\Vert f-F^\sharp_s(\dot{q}) \Vert_s\, \Vert \dot q\Vert_s
+g_B(1+x(s)^2) \Vert \dot{q} \Vert_s^2 }{g_B(E(q(s),\dot{q}(s)))\,
E(s,q(s),\dot{q}(s))},
\]
where $x(s):=\rho_s(p,q(s))$. Let us also define
\[
\Omega(x,y)=\frac{x y+x^2 g_B(1+x^2) + (1+x^2+y^2) \frac{y}{K+y}\,
g_B(1+x^2+y^2)
 +g_B(1+x^2) y^2}{(1+x^2+y^2) g_B(1+x^2+y^2)},
\]
and
\[
\beta =\sup_{x,y>0} 2\Omega(x,y)\le 6 <+\infty.
\]
 The bound on $f$ implies
$\Lambda(s)\le \Omega(\rho_s(p,q(s)), \Vert \dot q\Vert_s)$ thus
\begin{align*}
\pm [ E(t_2,q(t_2),\dot{q}(t_2))-&E(t_1,q(t_1),\dot{q}(t_1))]\le 2 \int_{t_1}^{t_2} \Omega(\rho_s(p,q(s)), \Vert \dot q\Vert_s) \, g_B(E) E\, \dd s\\
&\le \beta \int_{t_1}^{t_2}  \,  g_B(E(s,q(s),\dot{q}(s)))\,
E(s,q(s),\dot{q}(s)) \, \dd s.
\end{align*}
\end{proof}
%
%%\begin{align*}
%%E(t,q(t),\dot{q}(t))&\le \alpha +2 \int_{t_0}^t \Lambda(s) \, g_B(E) E\, \dd s\\
%%&\le \alpha+\beta \int_{t_0}^t  \,  g_B(E(q(s),\dot{q}(s)))\,
%%E(s,q(s),\dot{q}(s)) \, \dd s,
%%\end{align*}
%where $\alpha=1+ \rho_{t_0}(p,q(t_0))^2+\Vert \dot q \Vert_{t_0}^2$,
%and shortening $\zeta(s):=\rho_s(p,q(s))$, we have denoted
%\[
%\Lambda(s)=\frac{ \zeta(s) \Vert \dot q\Vert_s +g_B(1+\zeta(s)^2)\,
%\zeta(s)^2  + \Vert f-F^\sharp_s(\dot{q}) \Vert_s\, \Vert \dot
%q\Vert_s +g_B(1+\zeta(s)^2) \Vert \dot{q} \Vert_s^2
%}{g_B(E(q(s),\dot{q}(s)))\, E(s,q(s),\dot{q}(s))},
%\]
%%\[
%%\Lambda(s)=\frac{ \rho_s(p,q(s)) \Vert \dot q\Vert_s
%%+g(1+\rho_s(p,q(s))^2)\, \rho_s(p, q(s))^2  + \Vert
%%f-F^\sharp_s(\dot{q}) \Vert_s\, \Vert \dot q\Vert_s
%%+g_B(1+\rho_s(p,q(s))^2) \Vert \dot{q} \Vert_s^2   }{g_B(E)E}
%%\]
%and finally
%\begin{align*}
%\beta=&\sup_{q,v} \sup_{s\in [-B,B]} 2 \frac{[\rho(p,q) + \Vert
%f(s,q,v)-F^\sharp_t(v) \Vert] \, \Vert v \Vert }{g_B(E(q,v))E(q,v)}
%\\ &\le \sup_{q,v} \frac{2[\rho(p,q) \Vert v \Vert+
%\frac{g_B(E(q,v))}{K+\Vert v \Vert}\, E(q,v) \Vert v \Vert]
%}{g_B(E(q,v))E(q,v)} \le \frac{1}{c}+2<+\infty,
%%\\
%%&\le \sup_{q,v} \{1+ \frac{2 d(p,q) \Vert v
%%\Vert}{g_B(E(q,v))E(q,v)}\} \le 1+\frac{1}{c}<+\infty
%\end{align*}
%where $c>0$ is the lower bound for $g_B$.
%\sup_{q,v} \frac{2[d(p,q) + K_0(B)+K_1(B) d(p,q) +K_2(B) \Vert v
%\Vert + K_3(B) \frac{d(p,q)^2}{K_4(B) + \Vert v \Vert }] \, \Vert
%v \Vert }{E(q,v)}\\
%&\le 1+K_0(B)+K_1(B)+2K_2(B)+2K_3(B)<+\infty.
%\end{align*}

\begin{example} \label{cxe}
It is useful to study the differential equation on $\mathbb{R}$
given by \[\ddot q=\beta \frac{\dot{q}^2}{\alpha+q},\]
$\alpha,\beta>0$, also because of the analogy between this type of
force and the type of limiting behavior of  form
$K_3\frac{{q}^2}{K_4+\vert\dot{q}\vert}$ allowed by theorem
\ref{bso}.
 It is easy to check that for $\beta= 1$ it admits
complete solutions of the form $q(t)=c_2 \exp(c_1 t)-\alpha$, and
hence with any prescribed initial position $(q_0\ne -\alpha)$ and
velocity, while for $\beta=2$, the solutions are
$q(t)=\frac{c_1}{t-c_2}-\alpha$ and hence the equation admits a future
complete solution (i.e.\ $c_2<0$) only if $q(0)+\alpha$ and $\dot{q}(0)$
have opposite signs (consistently with \cite{hartman51}). This
example shows that if the force is quadratic in the velocity, we
need further assumptions in order to establish completeness, even if
the proportionality constant expressing such dependence is inversely
proportional to the distance.
\end{example}

\subsection{Application to Lagrangian mechanics and gravitational
 waves} \label{app}

Let ${Q}$ be a $d$-dimensional manifold (the space) endowed with the
(possibly time dependent) positive definite metric $a_t$, 1-form
field $b_t$ and potential function $V(t,q)$ (all $C^r$, $r \ge 2$).
On the classical spacetime $E=T \times {Q} $, where $T$ is a connected interval
of the real line, let $t$ be the time coordinate and let
$e_0=(t_0,q_0)$ and $e_1=(t_1,q_1)$ be events, the latter in the
future of the former i.e. $t_1>t_0$. Consider the action functional
of classical mechanics
\begin{equation} \label{cxs}
\mathcal{S}_{e_0,e_1}[q]=\int_{t_0}^{t_1} L(t,q(t), \dot{q}(t))
 \dd t,
\end{equation}
where
\begin{equation} \label{clas}
L(t,q, v)=\frac{1}{2}\, a_t(v,v) +b_t(v)-V(t,q),
\end{equation}
 on the space $C^1_{e_0,e_1}$ of $C^1$ curves $q\colon[t_0,t_1] \to
{Q}$ with fixed endpoints $q(t_0)=q_0$, $q(t_1)=q_1$.  Let $F_t:=d
b_t$, where $d$ is the exterior differentiation on $Q$ (thus $d$
does not differentiate with respect to $t$), that is, $F_t$ is the 2-form
whose components in local coordinates are $F_{t\, ij}=\partial_i
b_{t, j}-\partial_j b_{t, i}$. The $C^1$ stationary points are
smoother than the Lagrangian (namely $C^{r+1}$, see \cite[Theor.
1.2.4]{jost98}). By  Hamilton's principle, they solve the
Euler-Lagrange equation (e.g.\ \cite{eisenhart29} \cite[Eq.\
(2-39)]{greenwood77})
\begin{equation} \label{ele}
a_t(\cdot, \frac{D^{(t)}}{\dd t}\, \dot{q})= F_t(\cdot,
\dot{q})-(\p_t a_t)(\cdot,\dot{q})-( \p_t b_t+\p_q V),
\end{equation}
where,  as in previous sections, we denoted with $D^{(t)}$ the
affine connection of $a_t$ at the given time.

Historically this has proved to be one of the most important
variational problems because the mechanical systems of particles
subject to (possibly time dependent) holonomic constraints move
according to Hamilton's principle with a Lagrangian given by
(\ref{clas}) (see \cite{goldstein65}).

Theorem \ref{bso} allows us to establish the completeness of the
maximal solutions to the above Euler-Lagrange equations.

\begin{corollary} \label{bix}
Let ${Q}$ be a $d$-dimensional manifold (the space) endowed with the
(possibly time dependent) positive definite metric $a_t$, 1-form
field $b_t$ and potential function $V(t,q)$ (all $C^r$, $r \ge 2$).
Let us suppose that $(Q,a_t)$ are complete for each $t$ and let
 $\rho_t(p,q)$ be the corresponding Riemannian distance. Let us fix  $p\in Q$ and let us suppose that for every time interval $[-r,r]$ we can
find a continuous non-decreasing  function $g_r\colon [1,+\infty) \to
[1,+\infty)$ with the property that
\[
G_r(y):=\int_1^y \frac{1}{x g_r(x)}\, \dd x ,
\]
diverges for $y\to +\infty$ (the typical choice will be
$g_r=cnst.\ge 1$) and such that for every $(t,q,v)\in [-r,r]\times
TQ$
\begin{align}
\pm[ (\partial_t a_t)(v,v) ]&\le 2 g_r(1+\rho_t(p,q)^2) \,
a_t(v,v),\\
\Vert (\p_t b_t)(t,q) \Vert_t,\, \Vert (\p_q V)(t,q)\Vert_t &\le
g_r(1+\rho_t(p,q)^2) [1+\rho_t(p,q)],
 %\Vert f(t,q, v)-F^\sharp_t(v)\Vert_t &
%\le g_r(E(t,q,v))\,\frac{E(t,q,v)}{K+\Vert v \Vert_t},
\end{align}
 then the E.-L. flow is complete in the forward (resp.\ backward) direction.
\end{corollary}

\begin{corollary}Let ${Q}$ be a $d$-dimensional manifold (the space) endowed with the
(possibly time dependent) positive definite metric $a_t$, 1-form
field $b_t$ and potential function $V(t,q)$ (all $C^r$, $r \ge 2$).
If $Q$ is compact then the E.-L. flow of ($\ref{ele}$) is complete.
\end{corollary}

\begin{proof}
By the Hopf-Rinow theorem each Riemann space $(Q,a_t)$ is complete.
At any point $(t,q)\in [-r,r]\times Q$ we can choose a sufficiently
large constant $g_r$ so as to satisfy the inequalities of Cor.\
\ref{bix} for arbitrary $v$ at that point. Thus by compactness of
$[-r,r]\times Q$  we can find a sufficiently large constant $g_r$
such that the assumptions of that corollary hold.
\end{proof}

\begin{remark}
This Lagrangian problem has the following application to
gravitational waves. Let $M:=T\times {Q} \times \mathbb{R}$,
$T=\mathbb{R}$,  with $Q$ as above, and let an element of $M$ be
denoted by $(t,q,y)$. Let $M$ be endowed with the Lorentzian metric
\begin{equation} \label{eis}
g=a_t \!-\!\dd t \otimes (\dd y-\!b_t) -\!(\dd y-\!b_t) \otimes \dd
t-2V \dd t ^2 ,
\end{equation}
where the fields $a_t$, $b_t$ and $V$ are as above and the time
orientation is given by the global timelike vector
$W=-[V-\frac{1}{2}]\p_y+\p_t$, $g(W,W)=-1$. The future directed
lightlike vector  $n=\p_y$ can be shown to be covariantly constant.
In fact,  these spacetimes can be characterized as those spacetimes
which admit a covariantly constant null vector which generates a
$(\mathbb{R},+)$-principal fiber bundle.
 Spacetimes of this form are called {\em generalized gravitational
 waves} \cite{minguzzi06d} (more restrictive families are considered
 in \cite{beem96} and \cite{flores06}, the difference being
 essentially that passing between {\em general} and {\em natural} mechanical
 systems \cite{greenwood77}).

Eisenhart \cite{eisenhart29} realized that the {\em spacelike}
geodesics of $(M,g)$ project to $E:=T\times {Q}$ into solutions of
the above Euler-Lagrange equation and that any such solution can be
regarded as such projection. The author showed that the same can be
said with {\em spacelike} replaced by {\em lightlike}
\cite{minguzzi06d}, a fact particularly useful for its connection
with causality theory \cite{hawking73,minguzzi06d,minguzzi12f}. Thus
there is a very fruitful one-to-one correspondence between
generalized gravitational waves and rheonomic mechanical systems,
which allows one to import methods and ideas from one field to the
other \cite{minguzzi12f}.

For instance, under the assumption that $(Q,a_t)$ are complete, the
completeness of the E.-L. flow is equivalent to the geodesic
completeness of $(M,g)$, thus Corollary \ref{bix} establishes
conditions for the geodesic completeness of $(M,g)$. However, we
shall leave the details of this result to a different work
\cite{minguzzi12f}.
\end{remark}

% In the next immediate corollary (see also
%\cite{candela12} for similar results) $\Vert F_t^\sharp \Vert$ is
%the norm of the linear operator $F_t^\sharp: TQ_q\to TQ_q$, whose
%components are $a^{-1}{}^{ij} F_{jk}$. It can be useful to observe
%that $\Vert F_t^\sharp \Vert^2$ is bounded by $\textrm{Tr}
%F_t^2:=F_{t\,ij} F_t^{ij}$.

%Actually, under these assumptions the completeness of the E.-L. flow
%is equivalent to the geodesic completeness of $(M,g)$,

\section{Conclusions}
We have generalized Wintner's theorem on first order ODE flows over
$\mathbb{R}^n$ to Riemannian manifolds with possibly time-dependent
metrics. We have also given a second order version which is
particularly well shaped for application to Lagrangian mechanics.
The proofs are based on some inequalities for the Riemannian
distance which  allowed us to build non-smooth proper functions over
the manifold. The lack of smoothness was handled using the
LaSalle-Bihari generalization of Gronwall's inequality. Our results
can be easily generalized to second order inclusions so as to deal
with other interesting aspects of classical mechanics, such as
static friction. An application to the theory of exact gravitational
waves and lightlike dimensional reduction was also commented.

\section*{Acknowledgments}   
This work has been partially supported by GNFM of INDAM.

%\bibliography{../../bibliografie/simultaneity,../../bibliografie/libri,../../bibliografie/miei,../../bibliografie/mieiPreprints,../../bibliografie/mieiProceedings}

\begin{thebibliography}{10}
\providecommand{\url}[1]{\texttt{#1}}
\providecommand{\urlprefix}{URL } \expandafter\ifx\csname
urlstyle\endcsname\relax
  \providecommand{\doi}[1]{doi:\discretionary{}{}{}#1}\else
  \providecommand{\doi}{doi:\discretionary{}{}{}\begingroup
  \urlstyle{rm}\Url}\fi
\providecommand{\eprint}[2][]{\url{#2}}

\bibitem{wintner45}
A.~Wintner.
\newblock The non-local existence problem of ordinary differential equations.
\newblock \emph{Am. J. Math.}, \textbf{67} (1945) 277--284.

\bibitem{wintner46}
A.~Wintner.
\newblock The infinities in the non-local existence problem of ordinary
  differential equations.
\newblock \emph{Am. J. Math.}, \textbf{68} (1946) 173--178.

\bibitem{hartman64}
P.~Hartman.
\newblock \emph{Ordinary differential equations} (John {W}iley {\&} Sons, New
  York, 1964).

\bibitem{wintner46b}
A.~Wintner.
\newblock An {Abelian} lemma concernig asymptotic equilibria.
\newblock \emph{Am. J. Math.}, \textbf{68} (1946) 451--454.

\bibitem{cooke55}
K.~L. Cooke.
\newblock A non-local existence theorem for systems of ordinary differential
  equations.
\newblock \emph{Rend. Clrc. Matem. Palermo}, \textbf{4} (1955) 301--308.

\bibitem{conti56}
R.~Conti.
\newblock Limitazioni {``in ampiezza''} delle soluzioni di un sistema di
  equazioni differenziali e applicazioni.
\newblock \emph{Boll. Un. Mat. Ital.}, \textbf{11} (1956) 344–--349.

\bibitem{conti56b}
R.~Conti.
\newblock Sulla prolungabili{t\`a} delle soluzioni di un sistema di equazioni
  differenziali ordinarie.
\newblock \emph{Boll. Un. Mat. Ital.}, \textbf{11} (1956) 510--514.

\bibitem{strauss67}
A.~Strauss.
\newblock A note on a global existence result of {R. Conti}.
\newblock \emph{Boll. Un. Mat. Ital.}, \textbf{22} (1967) 434--441.

\bibitem{burton77}
T.~A. Burton.
\newblock A continuation result for differential equations.
\newblock \emph{Proc. Am. Math. Soc.}, \textbf{67} (1977) 272--276.

\bibitem{constantin91}
A.~Constantin.
\newblock Some observations on a {Conti's} result.
\newblock \emph{Rend. Mat. Acc. Lincei}, \textbf{2} (1991) 137--145.

\bibitem{gliklikh06}
Y.~E. Gliklikh.
\newblock Necessary and sufficient conditions for global in time existence of
  solutions.
\newblock \emph{Abstract and Applied Analysis}, \textbf{Article ID 39786}
  (2006) 1--17.

\bibitem{hartman51}
P.~Hartman, A.~Wintner.
\newblock On the non-increasing solutions of {$y'' = f(x,y,y')$}.
\newblock \emph{Am. J. Math.}, \textbf{73} (1951) 390--404.

\bibitem{mustafa02}
O.~G. Mustafa, Y.~V. Rogovchenko.
\newblock Global existence of solutions with prescribed asymptotic behavior for
  second-order nonlinear differential equations.
\newblock \emph{Nonlinear Analysis}, \textbf{51} (2002) 339--368.

\bibitem{ebin70}
D.~G. Ebin.
\newblock Completeness of {H}amiltonian vector fields.
\newblock \emph{Proc. Am. Math. Soc.}, \textbf{26} (1970) 632--634.

\bibitem{weinstein70}
A.~Weinstein, J.~Marsden.
\newblock A comparison theorem for {H}amiltonian vector fields.
\newblock \emph{Proc. Am. Math. Soc.}, \textbf{26} (1970) 629--631.

\bibitem{gordon70}
W.~B. Gordon.
\newblock On the completeness of {H}amiltonian vector fields.
\newblock \emph{Proc. Am. Math. Soc.}, \textbf{26} (1970) 329--331.

\bibitem{candela13}
A.~M. Candela, A.~Romero, M.~S\'anchez.
\newblock Completeness of the trajectories of particles coupled to a general
  force field.
\newblock \emph{Arch. Ration. Mech. Anal.}, \textbf{208} (2013) 255--274.

\bibitem{greenwood77}
D.~T. Greenwood.
\newblock \emph{Classical dynamics} (Dover, New York, 1977).

\bibitem{goldstein65}
H.~Goldstein.
\newblock \emph{Classical mechanics} ({Addison-Wesley} {P}ublishing {C}ompany,
  Reading, Massachusetts, 1965).

\bibitem{amann90}
H.~Amann.
\newblock \emph{Ordinary differential equations} (Walter de Gruyter, Berlin,
  1990).

\bibitem{hobson21}
E.~W. Hobson.
\newblock \emph{The theory of functions of a real variable and the theory of
  {F}ourier's series} (Cambridge University Press, Cambridge, 1921).

\bibitem{lasalle49}
J.~P. {LaSalle}.
\newblock Uniqueness theorems and successive approximations.
\newblock \emph{Ann. of Math.}, \textbf{50} (1949) 722--730.

\bibitem{bihari56}
I.~Bihari.
\newblock A generalization of a lemma of {B}ellman and its application to
  uniqueness problems of differential equations.
\newblock \emph{Acta Math. Acad. Scient. Hung.}, \textbf{7} (1956) 81--94.

\bibitem{hale80}
J.~K. Hale.
\newblock \emph{Ordinary differential equations} (Krieger Publishing Company,
  Malabar, Florida, 1980).

\bibitem{aubin84}
J.-P. Aubin, A.~Cellina.
\newblock \emph{Differential inclusions: Sets-valued maps and viability theory}
  ({Springer-Verlag}, Berlin, 1984).

\bibitem{jost98}
J.~Jost, X.~{Li-Jost}.
\newblock \emph{Calculus of Variations} (Cambridge {U}niversity {P}ress,
  Cambridge, 1998).

\bibitem{eisenhart29}
L.~P. Eisenhart.
\newblock Dynamical trajectories and geodesics.
\newblock \emph{Ann. Math. (Ser 2)}, \textbf{30} (1929) 591--606.

\bibitem{minguzzi06d}
E.~Minguzzi.
\newblock Eisenhart's theorem and the causal simplicity of {E}isenhart's
  spacetime.
\newblock \emph{Class. Quantum Grav.}, \textbf{24} (2007) 2781–--2807.

\bibitem{beem96}
J.~K. Beem, P.~E. Ehrlich, K.~L. Easley.
\newblock \emph{Global Lorentzian Geometry} (Marcel {D}ekker {I}nc., New York,
  1996).

\bibitem{flores06}
J.~L. Flores, M.~S{\'a}nchez.
\newblock On the geometry of pp-wave type spacetimes.
\newblock \emph{Lect. Notes Phys.}, \textbf{692} (2006) 79--9.

\bibitem{hawking73}
S.~W. Hawking, G.~F.~R. Ellis.
\newblock \emph{The Large Scale Structure of Space-Time} (Cambridge
  {U}niversity {P}ress, Cambridge, 1973).

\bibitem{minguzzi12f}
E.~Minguzzi.
\newblock Causality of spacetimes admitting a parallel null vector and weak
  {KAM} theory (2012).
\newblock ArXiv:1211.2685.

\end{thebibliography}
%\bibliographystyle{mio}

\end{document}